%% file: main.tex
\documentclass[letterpaper, 10 pt, journal, twoside]{IEEEtran}
\usepackage{graphicx}
\usepackage{epsfig}
\usepackage{times}
\usepackage{amsmath}
\usepackage{amsthm}
\usepackage{thmtools,thm-restate}
\usepackage{amssymb}
\usepackage[section]{placeins} 
\usepackage{placeins}
\usepackage{amsmath} 
\usepackage{multirow}
\usepackage{graphicx}
\usepackage[normalem]{ulem}
\usepackage{subcaption}
\usepackage{algorithm,tabularx}
\usepackage{algpseudocode}
\usepackage{amsfonts}
\usepackage{cases}
\usepackage{romannum}
\usepackage{textcomp}
\usepackage{xcolor}%
\usepackage{textcomp}
\providecommand{\U}[1]{\protect\rule{.1in}{.1in}}
\providecommand{\U}[1]{\protect\rule{.1in}{.1in}}

\IEEEoverridecommandlockouts
\newtheorem{assumption}{Assumption}
\newtheorem{theorem}{Theorem}

\newtheorem{lemma}{Lemma}

\newtheorem{remark}{Remark}
\newtheorem{definition}{Definition}

\newtheorem{property}{Property}

\useunder{\uline}{\ul}{}
\makeatletter
\newcommand{\multiline}[1]{  \begin{tabularx}{\dimexpr\linewidth-\ALG@thistlm}[t]{@{}X@{}}
#1
\end{tabularx}
}
\makeatother
\usepackage{cite}

\usepackage{picins}
\usepackage{centernot}

\usepackage{enumitem}
\setlist[itemize]{leftmargin=*}
\usepackage{makecell}
\usepackage[normalem]{ulem}
\usepackage{cuted}
\usepackage{hyperref}
\usepackage{stfloats}
\usepackage{lscape}
\usepackage{algorithm}
\usepackage{algpseudocode}

\newcommand{\tsup}[1]{\textsuperscript{#1}}
\newcommand{\T}{\top}
\newcommand{\I}{\mathbf{I}}
\newcommand{\0}{\mathbf{0}}

\newcommand{\bmtx}[1]{\begin{bmatrix}#1\end{bmatrix}}


\IEEEoverridecommandlockouts
\begin{document}

\title{\bf\Large Graph Neural Network-Based Distributed Optimal Control for Linear Networked Systems: An Online Distributed Training Approach}

\author{Zihao Song, Shirantha Welikala, Panos J. Antsaklis and Hai Lin\thanks{This work was supported by the National Science Foundation under Grant IIS-2007949. Zihao Song, Panos J. Antsaklis and Hai Lin are with the Department of Electrical Engineering, University of Notre Dame, Notre Dame, IN 46556 USA (e-mail: zsong2@nd.edu; pantsakl@nd.edu; hlin1@nd.edu;), and Shirantha Welikala is with the Department of Electrical and Computer Engineering, Stevens Institute of Technology, Hoboken, NJ 07030 USA (e-mail: swelikal@stevens.edu).}}
\maketitle
\thispagestyle{empty}

\begin{abstract}
   In this paper, we consider the distributed optimal control problem for linear networked systems. 
   In particular, we are interested in learning distributed optimal controllers using graph recurrent neural networks (GRNNs).
   Most of the existing approaches result in centralized optimal controllers with offline training processes.
   However, as the increasing demand of network resilience, the optimal controllers are further expected to be distributed, and are desirable to be trained in an online distributed fashion, which are also the main contributions of our work.
   % Our main contributions bridge the gaps of the GNN-based distributed optimal control by designing and learning a GRNN-based model in an online and distributed fashion.
   % {\color{red} 
   %   Most of the existing approaches ... design centrally... Our main contribution ... distributed ...
   % }
   %The main challenges lie in 1) the sparsity of the communication topologies between subsystems, 2) the scalability of the methods with respect to the network size, and 3) distributed online design. 
   To solve this problem, we first propose a GRNN-based distributed optimal control method, and we cast the problem as a self-supervised learning problem. Then, the distributed online training is achieved via distributed gradient computation, and inspired by the (consensus-based) distributed optimization idea, a distributed online training optimizer is designed. 
   Furthermore, the local closed-loop stability of the linear networked system under our proposed GRNN-based controller is provided by assuming that the nonlinear activation function of the GRNN-based controller is both local sector-bounded and slope-restricted.
   The effectiveness of our proposed method is illustrated by numerical simulations using a specifically developed simulator\footnote{Publicly available at \href{https://github.com/NDzsong2/fdTrainGRNN\textunderscore NetCtrl-main.git}{https://github.com/NDzsong2/fdTrainGRNN\textunderscore NetCtrl-main.git}}.
\end{abstract}
% \begin{IEEEkeywords}
%     Port-Hamiltonian, geometric control, rigid body, platoons, mesh stability.
% \end{IEEEkeywords}

%---------------------------------------------------------------
% \vspace{-2mm}
\section{Introduction}\label{sec:intro}
% \vspace{-1mm}

Designing distributed optimal controllers for networked systems has gained plenty of attention over the years due to the wide applications in many fields, such as intelligent transportation systems, power grid networks, and robotic networks. 
The main objective is to design a controller with a prescribed structure, as opposed to the traditional centralized controller, for a networked system consisting of an arbitrary number of interacting local subsystems. 
The main challenges of this problem lie in the sparsity of the communication topologies between subsystems and the scalability of the control methods with respect to the network size. As is shown in \cite{fardad2011sparsity,fattahi2018transformation}, the sparsity of the communication topologies may cause the non-convexity of the distributed optimal control problem or even make the problem intractable. Besides, as the network size grows, the performance of the controllers may not be preserved.

Existing works for the distributed optimal control problem focus on distributed optimization methods (e.g., consensus-based methods \cite{nedic2018distributed} and primal-dual decomposition \cite{lee2016distributed}), distributed model predictive control (DMPC) \cite{stewart2010cooperative}, multi-agent reinforcement learning (MARL) \cite{fan2023multi}, or other learning-based methods (e.g., a distinct multi-layer perceptron (MLP) \cite{farzanegan2024data} and adaptive critic control \cite{wang2020approximate}). 
Even though these methods admit convex formulations, they may usually result in complex solutions that do not scale with the size of the networks.

% {\color{red} Need a clear motivation on why choose a learning-based approach? Why GNN?}

Compared to traditional approaches, learning-based methods do not require the full model knowledge and provide more adaptability to complex systems and better robustness with the minimum communication overheads. 
Among the learning-based methods, Graph Neural Networks (GNNs) \cite{gama2022distributed}, as a special type of NNs, are especially suitable for the distributed optimal control problem, since the networked systems can be naturally modeled as graphs, where the subsystems are modeled as the nodes and the communication links are represented by edges. 
In particular, GNNs consist of a cascade of blocks (commonly known as layers), each of which applies a bank of graph filters followed by an element-wise nonlinearity to aggregate data from different nodes \cite{gama2022distributed} and can be trained inductively without the knowledge of the whole graph (as oppose to transductive training).
Moreover, GNNs are naturally local and distributed (in each layer), permutation equivariant and Lipschitz continuous to changes in the network.
These unique advantages make GNNs scalable (with respect to the computational complexity and graph size) and transferable. 

% In particular, GNNs are special types of neural networks that consist of a cascade of blocks (commonly known as layers), each of which applies a bank of graph filters followed by an element-wise nonlinearity to aggregate data from different nodes \cite{gama2022distributed}.
% Therefore, GNNs are suitable for distributed optimal control problem as compared to other learning-based methods.
So far, graph convolution neural networks (GCNNs) and graph recurrent neural networks (GRNNs) have been used to solve the distributed optimal control problem, as seen in \cite{yang2021communication,gama2022distributed}. 
However, the GNN models applied in these works are actually not distributed since the multi-layer convolution operation (as in GCNN models in \cite{gama2022distributed}) or the information propagation between the hidden states (as in GRNN models \cite{yang2021communication}) make the input data for each node depending not only on itself and its neighboring nodes but also on the neighbors' neighboring nodes.
The basic architecture and the information dependencies of the existing GCNN and GRNN models can be seen in Fig. \ref{Fig:existing_GNN_architectures}.
Despite recent progress in GNNs \cite{fallin2025lyapunov}, GNN-based controller design remains challenging, especially when the optimality is required. 
A comparison table is shown in Tab. \ref{tab:existing_methods_comparison} for existing methods.

\begin{table}[!h]
\caption{Comparison of different methods.}
\label{tab:existing_methods_comparison}
\begin{tabular}{c|c|c|c|c}
\hline
Method                                                   & \begin{tabular}[c]{@{}c@{}}Nonlinear \\ Systems\end{tabular}            & \begin{tabular}[c]{@{}c@{}}Non-convex \\ costs\end{tabular}          & \begin{tabular}[c]{@{}c@{}}Physical \\ Constr.\end{tabular}     & Scalability                                                                                                      \\ \hline
GCNNs                                                    & \begin{tabular}[c]{@{}c@{}}Yes (via\\ learning)\end{tabular}            & \begin{tabular}[c]{@{}c@{}}Possible \\ (in training)\end{tabular}   & \begin{tabular}[c]{@{}c@{}}Easily \\ involved\end{tabular}  & \begin{tabular}[c]{@{}c@{}} $O(N^2)$ or \\ $O(N^3)$ \end{tabular} \\ \hline
\begin{tabular}[c]{@{}c@{}}Primal\\ Decomp.\end{tabular} & \begin{tabular}[c]{@{}c@{}}Limited\\ (convex only)\end{tabular}         & \begin{tabular}[c]{@{}c@{}}No (requires \\ convexity)\end{tabular}   & \begin{tabular}[c]{@{}c@{}}Coupling \\ constraints\end{tabular} & \begin{tabular}[c]{@{}c@{}}$O(NM^2)$\end{tabular}                                   \\ \hline
\begin{tabular}[c]{@{}c@{}}Dual\\ Decomp.\end{tabular}   & \begin{tabular}[c]{@{}c@{}}Limited\\ (convex only)\end{tabular}         & \begin{tabular}[c]{@{}c@{}}No (strong \\ duality fails)\end{tabular} & \begin{tabular}[c]{@{}c@{}}Only soft \\ constraints\end{tabular}    & \begin{tabular}[c]{@{}c@{}}$O(NM^2)$ \\ $+$ delays\end{tabular}               \\ \hline
DMPC                                                     & \begin{tabular}[c]{@{}c@{}}Yes (nonlinear \\ MPC possible)\end{tabular} & \begin{tabular}[c]{@{}c@{}}Heuristic \\ only\end{tabular}            & \begin{tabular}[c]{@{}c@{}}Hard \\ constraints\end{tabular}  & \begin{tabular}[c]{@{}c@{}}$O(N\times$ \\ $T\times M^3)$\end{tabular}        \\ \hline
\begin{tabular}[c]{@{}c@{}}Our \\ Method\end{tabular}  & \begin{tabular}[c]{@{}c@{}}Yes (spatial \\ $+$ temporal) \end{tabular} & \begin{tabular}[c]{@{}c@{}}Yes \end{tabular}            & \begin{tabular}[c]{@{}c@{}}Penalties \\ or layers\end{tabular}  & \begin{tabular}[c]{@{}c@{}}$O(N)$\end{tabular}        \\ \hline
\end{tabular}
\end{table}

\begin{figure}[!t]
    \vspace{2mm}
    \centering
    \begin{subfigure}{0.4\textwidth}
        \includegraphics[width=\linewidth]{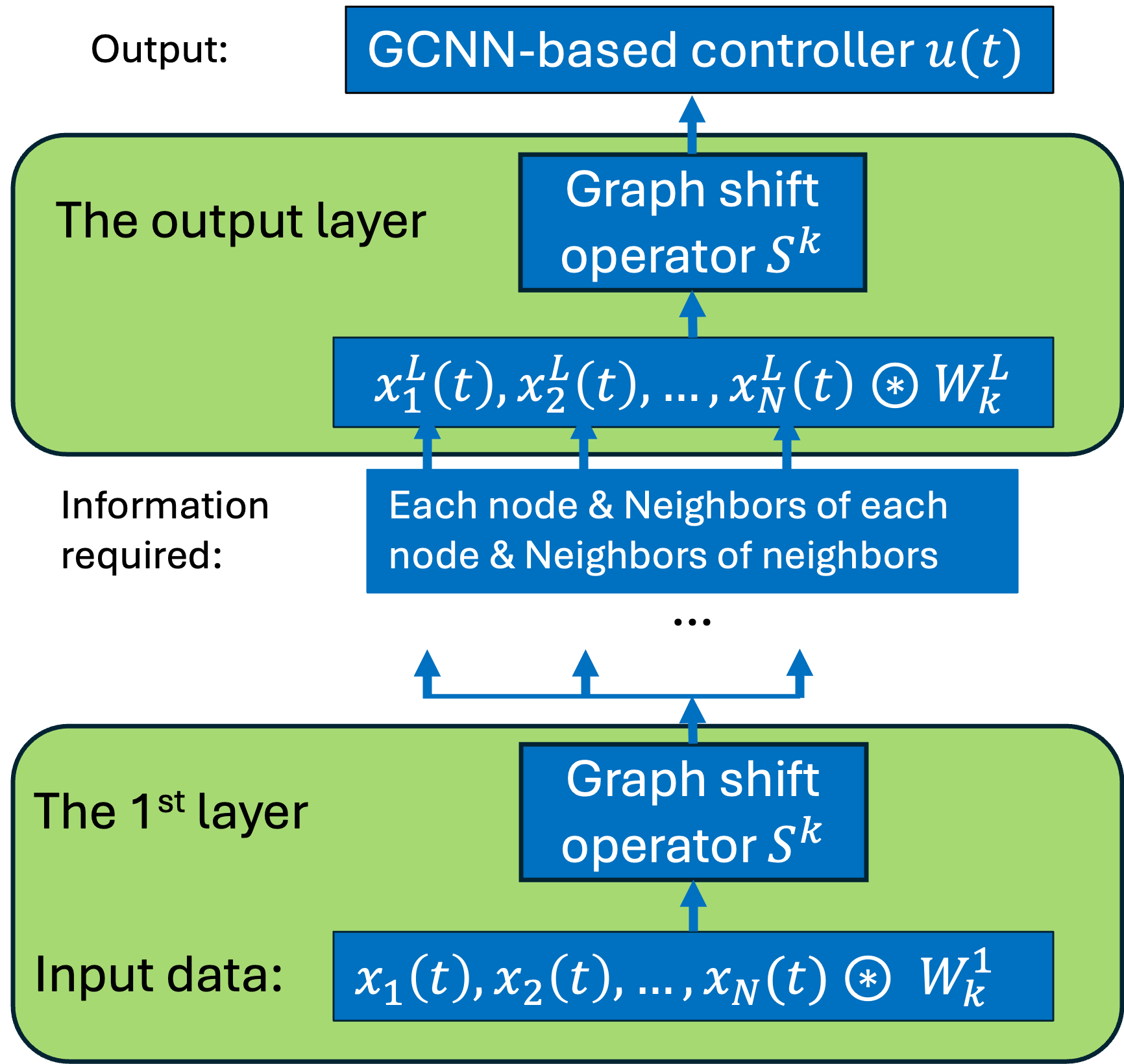}
        \vspace{-3.5mm}
        \caption{Existing works using GCNN-based controller \cite{gama2022distributed}}
        \label{Fig:existing_GCNN_architecture}
    \end{subfigure}
    \hfill
    \begin{subfigure}{0.4\textwidth}
        \includegraphics[width=\linewidth]{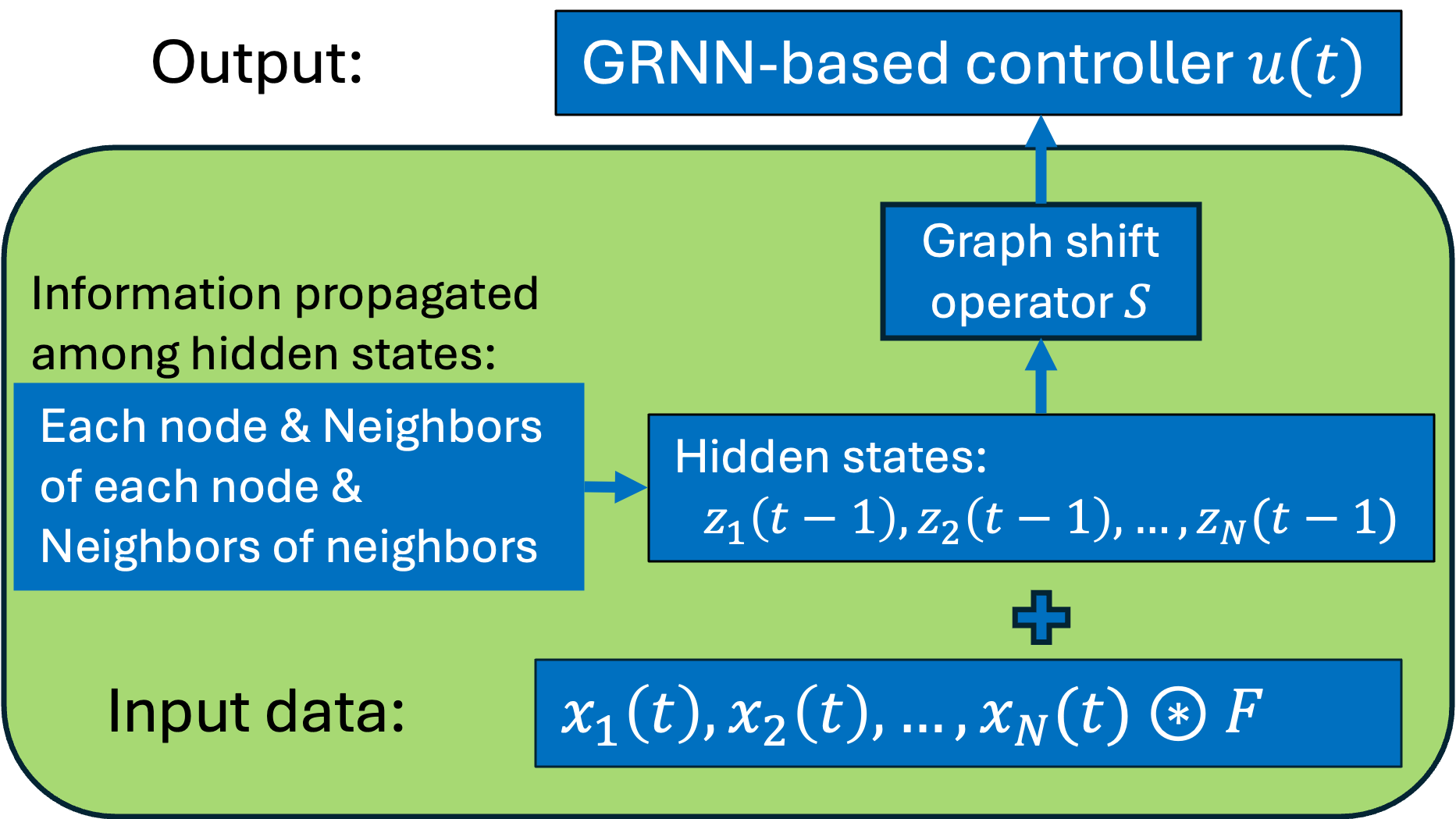}
        \vspace{-3.5mm}
        \caption{Existing works using GRNN-based controller \cite{yang2021communication}}
        \label{Fig:existing_GRNN_architecture}
    \end{subfigure}
    \caption{Architectures of the closed-loop networked system: (a) Existing works using GCNN-based controller \cite{gama2022distributed}; (b) existing works using GRNN-based controller \cite{yang2021communication}.}
    \label{Fig:existing_GNN_architectures}
\end{figure}

% \begin{figure}[!h]
%     \centering
%     \includegraphics[width=2.5in]{Figures/GCNN-based_controller_architecture.png}
%     \caption{Architectures of the closed-loop networked system: (a) Existing works using GCNN-based controller; (b) existing works using GRNN-based controller}
%     \label{Fig:existing_GCNN_architecture}
% \end{figure}

Apart from the distributed optimal control design, training the GNN-based controllers in an online distributed manner is another key point in this problem, as required by the increasing demand for network resilience. 
Even though there are some GNN-related works in network resilient tasks, they usually consider the problems such as edge rewiring \cite{yang2023learning}, critical node/link identification \cite{munikoti2022scalable}, and fault/attack detection \cite{baahmed2023using} and reinforcement learning-based control \cite{almasan2022deep}. Nevertheless, these methods do not consider the feedback dynamics structure, and the closed-loop stability of the entire network is generally not guaranteed.
Through online distributed training, the adaptivity of the controllers can be correspondingly enhanced to adapt to network perturbations, failures, or adversarial attacks so that the performance of the network will not deteriorate.

% {\color{gray} 
% Existing methods in network resilient control mainly rely on robustness-oriented methods, e.g., robust control methods (e.g., $H_{\infty}$ control \cite{hao2025neural} and fault-tolerant control \cite{li2024system}, secure control methods \cite{yue2021secure} and adversarial training for NN-based controllers \cite{khoshnevisan2023resilient}), focusing on designing systems that can withstand pre-defined disturbances or uncertainties without significant performance degradation. 
% On the other hand, the investigation of adaptivity-oriented methods, aiming at emphasizing the ability of systems to adapt dynamically to changing conditions or unforeseen disruptions, are relatively few. Most existing works use hierarchical control (using distributed control mechanisms to adapt at different levels) \cite{gusrialdi2022resilient}, game-theoretic approaches \cite{chen2019game}, adaptive control methods \cite{abbasi2025adaptive} and machine learning-based methods (e.g., reinforcement learning \cite{gilerson2025adaptive}, distributed learning methods \cite{zhao2025multi} and NN-based control).
% However, the approaches mentioned above are generally not scalable as the network size grows, especially for the learning-based approaches, where the training load will increase dramatically as more subsystems are involved. 
% }

% Inspired by \cite{olshevskyi2024fully},

Motivated by the above discussions, inspired by \cite{olshevskyi2024fully}, we aim to consider the distributed optimal control problem using a GRNN-based method, which is trained in an online distributed manner.
We first propose our GRNN-based controller for the distributed optimal control problem, and we cast this problem as a self-supervised learning problem. 
Then, the online distributed learning process is shown via the distributed gradient computation and online weights update law, where this update law is inspired by consensus-based distributed optimization. 
Furthermore, we prove the local closed-loop stability by assuming that the nonlinear activation function of the GRNN-based controller is both local sector-bounded and slope-restricted.
Eventually, the effectiveness of our results is verified by a specifically developed simulator. 

Our primary contributions can be summarized as follows:
\begin{enumerate}
    \item Different from the existing GNN-based methods \cite{yang2021communication,gama2022distributed}, which are actually not distributed, we propose a GRNN-based distributed optimal control method;
    \item Our GRNN-based distributed optimal controllers are trained in an online distributed fashion;
    \item The closed-loop stability of the networked system under our GRNN-based controller is proved by assuming the nonlinear activation function of our GRNN-based controller is both local sector-bounded and slope-restricted;
    \item The effectiveness of our proposed method is illustrated by a specifically developed simulator.
\end{enumerate}

The remainder of this paper is organized as follows. Some preliminaries and the problem formulation are presented in Section \ref{sec:background}. Our main results are presented in Section \ref{sec:main_results}, and are supported by a simulation example in Section \ref{sec:simulation}. Finally, concluding remarks are provided in Section \ref{sec:conclusion}.

%-------------------------------------------------------------
% \vspace{-2mm}
\section{Background}\label{sec:background}
% \vspace{-1mm}
%-------------------------------------------------------------
\paragraph*{\textbf{Notations}}

The sets of real, positive real, and natural numbers are denoted by $\mathbb{R}$, $\mathbb{R}_+$, and $\mathbb{N}$, respectively. $\mathbb{R}^{n\times m}$ denotes the space of real matrices with $n$ rows and $m$ columns. 
An $n$-dimensional real vector is denoted by $\mathbb{R}^n$. 
$\mathbb{R}^{n}_+$ denotes the sets that contain all component-wise nonnegative vectors in $\mathbb{R}^n$.
$\mathcal{I}_N:=\{1, 2,...,N\}$ is the index sets, where $N\in \mathbb{N}$. 
% $\mathbb{S}^n$ represent the sets of symmetric matrices in $\mathbb{R}^{n\times n}$.
The zero and identity matrices are denoted by $\0$ and $\I$, respectively (dimensions will be obvious from the context). 
A block matrix $A\in\mathbb{R}^{n\times m}$ is represented as $A:=[A_{ij}]_{i\in\mathcal{I}_n, j\in\mathcal{I}_m}$, where $A_{ij}$ is the $(i,j)$\tsup{th} block of $A$. $[A_{ij}]_{j\in \mathcal{I}_m}$ represents a block row matrix.
% and $\diag([A_{ii}]_{i\in\mathcal{I}_n})$ represents a block diagonal matrix. If the block $[A]_ii$ is scalar, then
$\mbox{diag}(a)\in\mathbb{R}^{n\times n}$ is a diagonal matrix generated by $a=[a_i]^\T\in\mathbb{R}^n$, where each scalar $a_i$ is on its diagonal (dimensions will be obvious from the context).
For a vector matrix $X=[x_i]_{i\in\mathcal{I}_N}^\T\in\mathbb{R}^{N\times n}$ with each $x_i\in\mathbb{R}^n$, we define $\mbox{vec}(X^\T):=[x_i^{\top}]_{i\in\mathcal{I}_N}^{\top}\in\mathbb{R}^{nN}$ as its vectorized form.  
For two vectors $x$, $y\in\mathbb{R}^n$, we use $x\preceq y$ ($x\succeq y$) to indicate that every element of $x-y$ is non-positive (non-negative), i.e., $x_i\leq y_i$ ($x_i\geq y_i$) for all $i\in \mathcal{I}_n$.
For matrix $A\in\mathbb{R}^{n\times m}$, its spectral norm is denoted by $\|A\|$.
For a vector $x\in\mathbb{R}^n$, its Euclidean norm is given by $|x|_2 := |x| = \sqrt{x^{\top}x}$. 
For a time dependent sequence $x(t):\mathbb{N}_0\rightarrow\mathbb{R}^n$, $\ell_2^n$ is the set of sequences $x(t)$ with $\|x(t)\|_2:=\sqrt{\sum_{\tau=0}^{\infty} x^\T(\tau)x(\tau)}<\infty$.    
% The space $l_2^p$ is the set of all real $p$-dimensional vectors with a finite $l_2$ norm.
% The $\mathcal{L}_2$ and $\mathcal{L}_{\infty}$ vectored function norms are given by $\|x(\cdot)\|=\sqrt{\int_{0}^{\infty}|x(t)|^2dt}$ and $\|x(\cdot)\|_{\infty} = \sup_{t\geq 0} |x(t)|$, respectively. $\mathcal{K}$ denotes class-$\mathcal{K}$ functions \cite{sontag1995characterizations}.

%----------------------------------------------------------
% \vspace{-1mm}
\subsection{Preliminaries}
% \vspace{-1mm}

Due to the element-wise nonlinear activation functions in our GRNN-based distributed optimal controllers, to show the stability of the closed-loop networked system, we need to approximate the output boundedness of the GRNN model. 
Therefore, we recall the following definitions of local sector and slope constraints. 

\begin{definition}(\textit{Local Sector-Bounded Nonlinearity \cite{wu2022stability}})\label{def:local_sector-bounded_nonlinearity}
    A nonlinear function $\sigma:\mathbb{R}\rightarrow\mathbb{R}$ is locally sector-bounded in $[\underline{\alpha},\ \bar{\alpha}]$ ($\underline{\alpha}\leq\bar{\alpha}\in\mathbb{R}$) on the interval $[\underline{\nu},\ \bar{\nu}]\subset \mathbb{R}$ if 
    $$(\sigma(\nu)-\underline{\alpha}\nu)(\sigma(\nu)-\bar{\alpha}\nu)\leq 0,\ \forall \nu\in[\underline{\nu},\ \bar{\nu}].$$
\end{definition}

\begin{definition}(\textit{Local Slope-Restricted Nonlinearity \cite{yin2021stability}})\label{def:local_slope-restricted_nonlinearity}
    A nonlinear function $\sigma:\mathbb{R}\rightarrow\mathbb{R}$ is locally slope-restricted in $[\underline{\sigma},\ \bar{\sigma}]$ ($\underline{\sigma}\leq\bar{\sigma}\in\mathbb{R}$) on the interval $[\underline{\nu},\ \bar{\nu}]\subset \mathbb{R}$ around the point ($\nu^*,\sigma(\nu^*)$), $\nu^*\in[\underline{\nu},\ \bar{\nu}]$, if 
    $$(\Delta\sigma(\nu)-\underline{\sigma}\Delta\nu)(\Delta\sigma(\nu)-\bar{\sigma}\Delta\nu)\leq 0,\ \forall \nu\in[\underline{\nu},\ \bar{\nu}],$$
    where $\Delta\sigma(\nu):=\sigma(\nu)-\sigma(\nu^*)$, and $\Delta\nu:=\nu-\nu^*$.
\end{definition}

\begin{remark}
    It is worth noting that the commonly used element-wise nonlinear activation functions, e.g., $\tanh$, (Leaky) ReLu, sigmoid, are all local sector bounded and local slope restricted \cite{yin2021stability}. 
    In particular, for some input $\nu\in\mathbb{R}$, $\tanh(\nu)$ and (Leaky) ReLu $\sigma(\nu):=\max(a\nu, \nu)$ with $a\in[0,1)$ ($a=0$ is only for Relu), are also global sector bounded and slope restricted, while for sigmoid activation function $\sigma(\nu):=\frac{1}{1+\exp(-\nu)}$,
    % and softmax $\sigma(\nu)=\frac{\exp(\nu)}{\sum_{} \exp(\nu_j)}$
    only local sector boundedness and slope restriction are guaranteed.
\end{remark}

Based on the above definitions, we further recall the time domain integral quadratic constraints (IQCs) for stability analysis of the discrete-time systems with NN controllers.
\begin{definition}(\textit{Time Domain IQCs \cite{wu2022stability}})\label{Def:time_domain_IQC}
    Suppose that a set $\mathcal{Q}$ is the positive semi-definite cone, where any $Q\in\mathcal{Q}$ is a symmetric matrix. 
    Consider a discrete-time LTI system $\Psi:\ell_2^p\times\ell_2^p\rightarrow\ell_2^r$ as
    \begin{equation}
        \begin{split}
            \psi(t+1) &= A_{\psi} \psi(t)+ B^h_{\psi} h(t)+ B^z_{\psi} z(t),       \\
            q(t) &= C_{\psi} \psi(t)+ D^h_{\psi} h(t)+ D^z_{\psi} z(t),
        \end{split}
    \end{equation}
    where $\psi(0)=\0$. 
    Two sequences $h$, $z\in\ell_2^p$ satisfy the IQC defined by ($\Psi,\mathcal{Q}$), if the output sequence $q\in\ell_2^r$ satisfies
    \begin{equation}
        \sum_{\tau=0}^{t} q^\T(\tau)Qq(\tau)\geq 0,
    \end{equation}
    for any $Q\in\mathcal{Q}$, and $t\in\mathbb{N}$, and $\Psi:=\scriptsize\bmtx{\begin{array}{c|c c} A_{\Psi} & B^h_{\Psi} & B^z_{\Psi} \\  \hline 
            C_{\Psi} & D^h_{\Psi} & D^z_{\Psi}
            \end{array}}$.
\end{definition}

% use the following definition of Input-state stability:
% \begin{definition}(\textit{Input-State Stability \cite{gama2022distributed}})\label{Def:input-state_stability}
%     {\color{blue}Consider a linear networked system as in \eqref{Eq:stacked_networked_system} controlled by $U(t) = \Phi(X(t))+E(t)$ where $E(t)$ is a disturbance term or exploratory signal. The system is input-state stable if, for all sequences $\{X(t)\}$ and $\{E(t)\}$ such that $\sum_{t=0}^{\infty} \|X(t)\|<\infty$ and $\sum_{t=0}^{\infty} \|E(t)\|<\infty$, there exist constants $\beta_0$, $\beta_1\geq 0$ such that:
%     \begin{equation}
%         \sum_{t=0}^{\infty} \|X(t)\|\leq \beta_0+\beta_1\sum_{t=0}^{\infty} \|E(t)\|.
%     \end{equation}
%     }
% \end{definition}

%----------------------------------------------------------
% \vspace{-1mm}
\subsection{Problem Formulation}

Consider a linear networked system $\Sigma$ with the $i\tsup{th}$ subsystem $\Sigma_i$ (for all $i\in\mathcal{I}_N$) being:
\begin{equation}\label{Eq:ith_subsystem}
    \Sigma_i: x_i(t+1) = \sum_{j\in\mathcal{I}_N} A_{ij}x_j(t)+B_{ij}u_j(t),
\end{equation}
where $x_i\in\mathbb{R}^{n}$ is the system states, $u_i\in\mathbb{R}^{m}$ is the system input,
% $d_i\in\mathbb{R}^n$ is the time-varying bounded external disturbances.
$A_{ij}\in\mathbb{R}^{n\times n}$ is the system matrix and $B_{ij}\in\mathbb{R}^{n\times m}$ is the input mapping matrix. 
Here, the interconnection topologies between subsystems in \eqref{Eq:ith_subsystem} are assumed as directed graphs $G(\mathcal{V},\mathcal{E})$, where $\mathcal{V}\equiv\mathcal{I}_N$ is the node set, and $\mathcal{E}\subseteq \mathcal{V}\times\mathcal{V}$ is the edge set.
Here, the $i\tsup{th}$ node in the graph $G$ represents the $i\tsup{th}$ subsystem in \eqref{Eq:ith_subsystem}. If node $i$ can receive information from node $j$, $A_{ij}$, $B_{ij}\neq \0$; otherwise, $A_{ij}$, $B_{ij}= \0$.

If we stack all the subsystems $\Sigma_i$ ($i\in\mathcal{I}_N$) into a vector form, the linear networked system $\Sigma$ becomes:
\begin{equation}\label{Eq:stacked_networked_system}
    \Sigma: \mbox{vec}(X^\T(t+1)) = A\mbox{vec}(X^\T(t))+B\mbox{vec}(U^\T(t)),
\end{equation}
where $A=[A_{ij}]_{i,j\in\mathcal{I}_N}\in\mathbb{R}^{nN\times nN}$ and $B=[B_{ij}]_{i,j\in\mathcal{I}_N}\in\mathbb{R}^{nN\times mN}$ are the stacked system matrix and the stacked input mapping matrix, respectively. $X(t)=[x_i(t)]_{i\in\mathcal{I}_N}^\T\in\mathbb{R}^{N\times n}$ and $U(t) = [u_i(t)]_{i\in\mathcal{I}_N}^\T\in\mathbb{R}^{N\times m}$ are the stacked state vector and the stacked control input vector, respectively.
% $d(t):=[d_i^\T]^\T_{i\in\mathcal{I}_N}\in\mathbb{R}^{nN}$ is the stacked external disturbances vector.
For the networked system $\Sigma$ in \eqref{Eq:stacked_networked_system}, we assume that the pair $(A,B)$ is controllable. 
Moreover, we assume that there is no time delay between the transmission between subsystems.
% with $n:=\sum_{i\in\mathcal{I}_N} n_i$ and $m:=\sum_{i\in\mathcal{I}_N} m_i$, respectively.

Then, we can formulate the GRNN-based distributed optimal control problem as:
\begin{subequations}\label{Eq:distributed_optimal_control_original}
    \begin{align}
        \min_{\phi\in\Phi}\ \ & J\big(\{X(t)\}_{t=0}^{\infty},\ \{U(t)\}_{t=0}^{\infty}\big)  \label{Eq:centralized_cost}  \\
        s.t.\ & \eqref{Eq:stacked_networked_system},\ \forall t\in\{0,1,...\},   \nonumber  \\
        & U(t) = \phi(X(t);G),\ \forall t\in\{0,1,...\}.  \label{Eq:distributed_optimal_controller}
    \end{align}
\end{subequations}

The objective of this paper is to learn a GRNN-based distributed optimal controller as in \eqref{Eq:distributed_optimal_controller} in an online distributed fashion for the networked system $\Sigma$ in \eqref{Eq:stacked_networked_system}, such that the closed-loop networked system is stable.

\begin{remark}
    For the problem in \eqref{Eq:distributed_optimal_control_original}, traditional model-based approaches are either centralized or not scalable, due to the sparsity constraints caused by the interconnection topologies and the dramatic increase of the computational complexity as the network size grows, respectively.    
    In addition, as illustrated in Section \ref{sec:intro}, existing GCNN- and GRNN-based control methods are not distributed (as seen in Fig. \ref{Fig:existing_GNN_architectures}), which is not desirable for the current network resilience requirements.
    Motivated by these observations, we choose to design a GRNN-based distributed control method for the problem in \eqref{Eq:distributed_optimal_control_original}, where the hidden states possess dynamical structure to adapt to the changes in the networks.
    To make the controller distributed, no multi-layer convolutional operations for the input data or information propagation in hidden states are allowed in our model.
    % {\color{red} Difficulties for classical model-based approaches, motivations for the choice of GRNN?}
\end{remark}

\begin{remark}
    From the individual subsystem $\Sigma_i$ in \eqref{Eq:ith_subsystem} and our control objectives, it is worth noting that the interconnection topology applied in the distributed controller $U(t)$ is free to be selected from a subset of the interconnection links in \eqref{Eq:stacked_networked_system}.
\end{remark}

%-----------------------------------------------------------------
\section{Main Results}\label{sec:main_results}

In this section, we propose an online distributed training method for the GRNN-based distributed optimal controller. 
We first introduce the GRNN structure (to be further trained as our distributed optimal controller). Then, we show the algorithm for training the GRNN-based distributed optimal controller in an online distributed manner. Finally, the closed-loop stability is proved for the networked system $\Sigma$ in \eqref{Eq:stacked_networked_system} under our learned controller.

%-----------------------------------------------------------------------------
\subsection{Online Distributed Learning for GRNN-based Distributed Optimal Controller}\label{subsec:distributedGCNN}

In the case of the centralized training, the GRNN-based distributed optimal controller $U(t)$ is designed as:
\begin{align}\label{Eq:distributed_optimal_controller_GRNN}
    Z(t) &= \sigma\Big(Z(t-1)\Theta_1+X(t)\Theta_2+SX(t)\Theta_3\Big),  \nonumber    \\
    U(t) &= Z(t)\Theta_4,   
\end{align}
where $X(t)\in\mathbb{R}^{N\times n}$ is the input data of the GRNN that collects the real-time states of the networked system \eqref{Eq:stacked_networked_system}, for all $t\in\{0,1,...\}$. 
$Z(t)\in\mathbb{R}^{N\times p}$ is the full internal state of the GRNN by stacking all the subsystems' internal states as $Z(t):=[z_i(t)]_{i\in\mathcal{I}_N}^\T$  with each $z_i(t)\in\mathbb{R}^p$, for all $t\in\{0,1,...\}$. 
Here, $\Theta_1\in\mathbb{R}^{p\times p}$, $\Theta_4\in\mathbb{R}^{p\times m}$, and $\Theta_2$, $\Theta_3\in\mathbb{R}^{n\times p}$ are the trainable parameters (or weights), $\sigma$ is an element-wise nonlinear activation function, and $S\in\mathbb{R}^{N\times N}$ is the graph shift operator, characterizing the interconnection topologies between neighboring subsystems for the controller \eqref{Eq:distributed_optimal_controller_GRNN}. 
For the graph shift operator $S:=[S_{ij}]_{i,j\in\mathcal{I}_N}$, the element $S_{ij}$ can be non-zero iff $i=j$ or $j\in\mathcal{N}_i:=\{i:(v_j,v_i)\in\mathcal{E}\}$, and is zero, otherwise.
$S$ can be selected as the (normalized) adjacency matrix or the (normalized) Laplacian matrix of the interconnection topological graph.

Note that even though the controller in \eqref{Eq:distributed_optimal_controller_GRNN} can be implemented in a distributed manner towards all subsystems, it is typically trained in a centralized manner, since the entire input data for all nodes are required during training and the weights are shared among all the nodes. 
In this way, it cannot adjust the weights for each subsystems online. Thus, the resulting controller \eqref{Eq:distributed_optimal_controller_GRNN} follow a paradigm of 'centralized training, distributed execution'.

To enable online distributed adjustment of the weights of the GRNN-based distributed optimal controller using only local information, we design the local form of the GRNN-based distributed optimal controller at the $i\tsup{th}$ subsystem as:
\begin{subequations}\label{Eq:distributed_optimal_control_GRNN_local}
    \begin{align}
        Z_{[i,:]}(t) &= \sigma(h_{[i,:]}(t)),  \label{Eq:Zt_[i*]}  \\
        h_{[i,:]}(t) &= Z_{[i,:]}(t-1)\Theta_{1i}(t)+X_{[i,:]}(t)\Theta_{2i}(t)+  \nonumber \\
        & \sum_{j\in\mathcal{N}_i\cup\{i\}} S_{ij}X_{[j,:]}(t)\Theta_{3i}(t),  \label{Eq:ht_[i*]}    \\
        U_{[i,:]}(t) &= Z_{[i,:]}(t)\Theta_{4i}(t),    \label{Eq:ut_[i*]}
    \end{align}
\end{subequations}
for any time $t\in\{0,1,...\}$, where $(\cdot)_{[i,:]}(t)$ represent the $i\tsup{th}$ row of the matrix $(\cdot)(t)$, and thus, the GRNN model in \eqref{Eq:Zt_[i*]}-\eqref{Eq:ut_[i*]} depends only on the states of the $i\tsup{th}$ subsystem itself and its neighboring subsystems, and the weights of the $i\tsup{th}$ subsystem $\Theta_{1i}$, $\Theta_{2i}$, $\Theta_{3i}$ and $\Theta_{4i}$. 
Initially, $Z_{[i,:]}(-1) = \0$, and $\Theta_{1i}(0)$, $\Theta_{2i}(0)$, $\Theta_{3i}(0)$ and $\Theta_{4i}(0)$ are the local copies of the centralized weights $\Theta_1$, $\Theta_2$, $\Theta_{3}$ and $\Theta_{4}$ in \eqref{Eq:distributed_optimal_controller_GRNN} at the $i\tsup{th}$ subsystem, respectively, i.e., $\Theta_{1i}(0)=\Theta_1$, $\Theta_{2i}(0)=\Theta_2$, $\Theta_{3i}(0)=\Theta_3$ and $\Theta_{4i}(0)=\Theta_4$, for all $i\in\mathcal{I}_N$ and $j\in\mathcal{N}_i$. 
During the online training, these local copies can be updated separately and will not be the local copies of the centralized weights anymore, and the centralized weights in \eqref{Eq:distributed_optimal_controller_GRNN} can be approximated by $\Theta_k=\frac{1}{N}\sum_{i\in\mathcal{I}_N}\Theta_{ki}$, for all $k\in\mathcal{I}_4$.

\begin{remark}
    The physical interpretation of the local structure of the GRNN-based distributed optimal controller in \eqref{Eq:distributed_optimal_control_GRNN_local} is that each layer follows the form of the classical distributed controller, e.g., $u_i=\sum_{j\in\mathcal{N}_i\cup\{i\}} K_{ij}x_j$, for all $i\in\mathcal{I}_N$, capturing the states/features of the $i\tsup{th}$ subsystem and its neighboring subsystems $j\in\mathcal{N}_{i}$. Rather than assigning a local weight $\Theta_{3ij}$ for each neighboring information $X_{[j,:]}(t)$, for all $i\in\mathcal{I}_N$ and $j\in\mathcal{N}_i$, we simply learn a single $\Theta_{3i}$ at each subsystem so as to improve the scalability.
\end{remark}

To achieve online distributed training for our GRNN-based controllers \eqref{Eq:distributed_optimal_control_GRNN_local}, we make the following assumption.
\begin{assumption}\label{Asm:distributed_optimization_cost}
    The total cost $J$ in \eqref{Eq:centralized_cost} can be equivalently constructed by the summation of the decoupled sub-costs $J_i$ for all $i\in\mathcal{I}_N$, i.e., $J\equiv\sum_{i\in\mathcal{I}_N} J_i$.
\end{assumption}

Based on Asm. \ref{Asm:distributed_optimization_cost}, we can reformulate the original GRNN-based distributed optimal control problem in \eqref{Eq:distributed_optimal_control_original} as:
\begin{subequations}\label{Eq:distributed_optimal_control_local}
    \begin{align}
        \min_{\sigma,\ \{\Theta_{ki}: i\in\mathcal{I}_N, k\in\mathcal{I}_4\}}\ \ & \sum_{i\in\mathcal{I}_N} J_i\big(\{X_{[i,:]}(t)\}_{t=0}^{\infty},\ \{U_{[i,:]}(t)\}_{t=0}^{\infty}\big) \label{Eq:distributed_cost}   \\
        s.t.\ \ & \eqref{Eq:ith_subsystem},\ \eqref{Eq:distributed_optimal_control_GRNN_local},\ \forall i\in\mathcal{I}_N,\ j\in\mathcal{N}_i,      \\
        & \mbox{and } \forall t\in\{0,1,...\},    \nonumber
    \end{align}
\end{subequations}
where the $i\tsup{th}$ subsystem exchanges its input node states/features with its neighboring subsystems $j\in\mathcal{N}_i$ as in \eqref{Eq:ith_subsystem} and locally trains \eqref{Eq:distributed_optimal_control_GRNN_local}.

The distributed gradient update of the local GRNN model \eqref{Eq:distributed_optimal_controller_GRNN} follows the traditional backpropagation process, which is a message passing process. Therefore, for distributed training, we need to ensure that the messages passed are also distributed during training over all the subsystems and layers, i.e., the gradient of each weight in the corresponding subsystem only depends on local information.

To show the distributed backpropagation process, we find the relationship between the weights and the corresponding sub-cost $J_i$ by noting that only the control input $U_{[i,:]}(t)$ is related with all the weights $\Theta_{ki}$, for all $k\in\mathcal{I}_4$. Therefore, we have (the time dependencies are omitted here, but we use superscript when necessary):
\begin{subequations}\label{Eq:distributed_gradients}
    \begin{align}
        \frac{\partial J}{\partial \Theta_{1i}} &=\frac{\partial J_i}{\partial \Theta_{1i}} = \frac{\partial J_i}{\partial U_{[i,:]}}\frac{\partial U_{[i,:]}}{\partial Z_{[i,:]}}\frac{\partial Z_{[i,:]}}{\partial h_{[i,:]}}\frac{\partial h_{[i,:]}}{\partial \Theta_{1i}}     \nonumber  \\
        &=\Big(\frac{\partial J_i}{\partial U_{[i,:]}}\Theta_{4i}^\T\Big)\odot\sigma'(h_{[i,:]})Z_{[i,:]}^{(t-1)}
        \label{Eq:local_gradient_J_to_Theta1i}   \\
        \frac{\partial J}{\partial \Theta_{2i}} &=\frac{\partial J_i}{\partial \Theta_{2i}} = \frac{\partial J_i}{\partial U_{[i,:]}}\frac{\partial U_{[i,:]}}{\partial Z_{[i,:]}}\frac{\partial Z_{[i,:]}}{\partial h_{[i,:]}}\frac{\partial h_{[i,:]}}{\partial \Theta_{2i}}         \nonumber   \\
        &= \Big(\frac{\partial J_i}{\partial U_{[i,:]}}\Theta_{4i}^\T\Big)\odot\sigma'(h_{[i,:]})X_{[i,:]},
        \label{Eq:local_gradient_J_to_Theta2i}    \\
        \frac{\partial J}{\partial \Theta_{3i}} &= \frac{\partial J_i}{\partial \Theta_{3i}} = \frac{\partial J_i}{\partial U_{[i,:]}}\frac{\partial U_{[i,:]}}{\partial Z_{[i,:]}}\frac{\partial Z_{[i,:]}}{\partial h_{[i,:]}}\frac{\partial h_{[i,:]}}{\partial \Theta_{3i}}         \nonumber   \\
        &= \Big(\frac{\partial J_i}{\partial U_{[i,:]}}\Theta_{4i}^\T\Big)\odot\sigma'(h_{[i,:]})\sum_{j\in\mathcal{N}_i\cup\{i\}} S_{ij}X_{[j,:]},
        \label{Eq:local_gradient_J_to_Theta3i}    \\
        \frac{\partial J}{\partial \Theta_{4i}} &= \frac{\partial J_i}{\partial \Theta_{4i}} = \frac{\partial J_i}{\partial U_{[i,:]}}\frac{\partial U_{[i,:]}}{\partial \Theta_{4i}} = \frac{\partial J_i}{\partial U_{[i,:]}}Z_{[i,:]},
        \label{Eq:local_gradient_J_to_Theta4i}
    \end{align}
\end{subequations}
where $\odot$ is the Hadamard product, i.e., the element-wise product between vectors. In other words, for $A=[a_{ij}]\in\mathbb{R}^{m\times n}$ and $B=[b_{ij}]\in\mathbb{R}^{m\times n}$, $[A\odot B]_{ij}:=a_{ij}\cdot b_{ij}$, for all $i\in\mathcal{I}_m$ and $j\in\mathcal{I}_n$.

To online update the weights $\Theta_{ki}$, for all $i\in\mathcal{I}_N, k\in\mathcal{I}_4$ of the GRNN-based controller in \eqref{Eq:distributed_optimal_control_GRNN_local}, inspired by the idea of distributed optimization \cite{nedic2020distributed}, we propose a consensus-based weight updating approach.
In particular, we first decentrally update the weights for each subsystem, and then, we find the consensus solution by taking the neighboring information into account.
Thus, the update law for each weight at the $i\tsup{th}$ subsystem is:
\begin{subequations}\label{Eq:D-SGD_optimizer}
    \begin{align}
        \theta_i(t+1) &= \theta_i(t) - \eta_t\frac{1}{n_B}\sum_{b=1}^{n_B}\nabla J_{bi}(\theta_i(t)),      \\
        \theta_i(t+1) &= \sum_{j\in\mathcal{N}_i\cup\{i\}} W_{ij}\theta_j(t+1),
    \end{align}
\end{subequations}
where $\theta_i\in\{\Theta_{ki}$, for all $i\in\mathcal{I}_N, k\in\mathcal{I}_4\}$,
$\eta_t$ is the variable learning rate,
$n_B$ is the batch size in one round of computation, and the consensus matrices $W_{ij}$'s are selected as the Metropolis-Hasting weights:
\begin{equation}
    W_{ij}:=\begin{cases}
        \frac{1}{\max\{d_i, d_j\}},    & \mbox{for } j\in\mathcal{N}_i,   \\
        1-\sum_{k\in\mathcal{N}_i} W_{ik},  & \mbox{for } i = j,    \\
        0,    & \mbox{otherwise},
    \end{cases}
\end{equation}
where $d_i$ and $d_j$ represent the degree of node $i$ and $j$, respectively.

Based on the above results, the procedure of our proposed (consensus-based) online distributed training for the GRNN-based distributed optimal controller \eqref{Eq:distributed_optimal_control_GRNN_local} to solve the distributed optimal control problem in \eqref{Eq:distributed_optimal_control_local} can be summarized in Alg. \ref{alg:distributed_optimal_control_algorithm}.

% on a network of a fixed set of vertices V with dynamic topology, i.e., edge set E(t), is illustrated in Algorithm 1, where Ω is the sampling distribution for (G,X,y), and ΩV (or ΩG) is the conditional distribution for a given V (or G).

\begin{algorithm}
\caption{Online Distributed Training for GRNN-Based Distributed Optimal Control \eqref{Eq:distributed_optimal_control_local}}
\label{alg:distributed_optimal_control_algorithm}
\begin{algorithmic}[1]
    \State \textbf{Input:} Initial states $\{X_{[i,:]}(0):i\in\mathcal{I}_N\}$, initial weights $\{\Theta_{ki}(0)$: $k\in\mathcal{I}_4, i\in\mathcal{I}_N\}$, initial hidden states $\{Z_{[i,:]}(-1)=\0:i\in\mathcal{I}_N\}$
    \State \textbf{Output:} Local weights $\{\Theta^*_{ki}: k\in\mathcal{I}_4, i\in\mathcal{I}_N\}$;
    \State Set up the topology $S$ and initial random weights $\Theta_{ki}$;
    \State Set up the D-SGD optimizer \eqref{Eq:D-SGD_optimizer};
    \For{each epoch in TotalEpochs}
        \State Update states $X_{[i,:]}(t+\Delta T)$ and control input $U_{[i,:]}(t+\Delta T)$ over a small time horizon $\Delta T$;
        \State Compute the loss \eqref{Eq:distributed_cost};
        \State Compute the local distributed gradients \eqref{Eq:distributed_gradients} for each node $i$ individually $i\in\mathcal{I}_N$;
        
        \State Update $\Theta_{ki}(t)$ by \eqref{Eq:D-SGD_optimizer} and the computed local gradients for each node;
        \State Compute the testing loss with the updated weights;
    \EndFor
    \State \Return Trained local weights $\Theta^*_{ki}$, training and testing losses.
\end{algorithmic}
\end{algorithm}

%-----------------------------------------------------------------
\subsection{Stability Analysis}\label{subsec:stability_analysis}

In this part, we analyze the stability of the closed-loop networked system \eqref{Eq:stacked_networked_system} with our GRNN-based distributed optimal controller \eqref{Eq:distributed_optimal_control_GRNN_local}.

First, note that the networked system $\Sigma$ in \eqref{Eq:stacked_networked_system} can be equivalently rewritten as:
\begin{equation}\label{Eq:stacked_networked_system_rewritten}
    \Sigma:\ x(t+1) = Ax(t)+Bu(t),
\end{equation}
where the states and control inputs are $x(t)\equiv\mbox{vec}(X^\T(t))\in\mathbb{R}^{nN}$ and $u(t)\equiv\mbox{vec}(U^\T(t))\in\mathbb{R}^{mN}$, respectively.

Based on \eqref{Eq:stacked_networked_system_rewritten} and \eqref{Eq:distributed_optimal_control_GRNN_local}, the closed-loop networked system with our GRNN-based controller can be written as: 
\begin{equation}\label{Eq:closed-loop_GRNN_feedback}
    \begin{cases}
        x(t+1) = Ax(t)+BK_4 z(t),    \\
        h(t) = K_1 z(t-1)+(K_2+K_3 S)x(t),  \\
        z(t) = \sigma(h(t)),
    \end{cases}
\end{equation}
where $K_{k}:=\mbox{diag}\{\Theta_{ki}^\T:k\in\mathcal{I}_4,\ i\in\mathcal{I}_N\}$ are the collection fo local trainable weights, and 
$h(t)\in\mathbb{R}^{pN}$ and $z(t)\equiv\mbox{vec}(Z^\T(t))\in\mathbb{R}^{pN}$ are the hidden states and internal states, respectively. The architecture of the closed-loop dynamics of \eqref{Eq:closed-loop_GRNN_feedback} is shown in Fig. \ref{Fig:closed-loop_architecture}.

\begin{figure}[!h]
    \centering
    \includegraphics[width=1.8in]{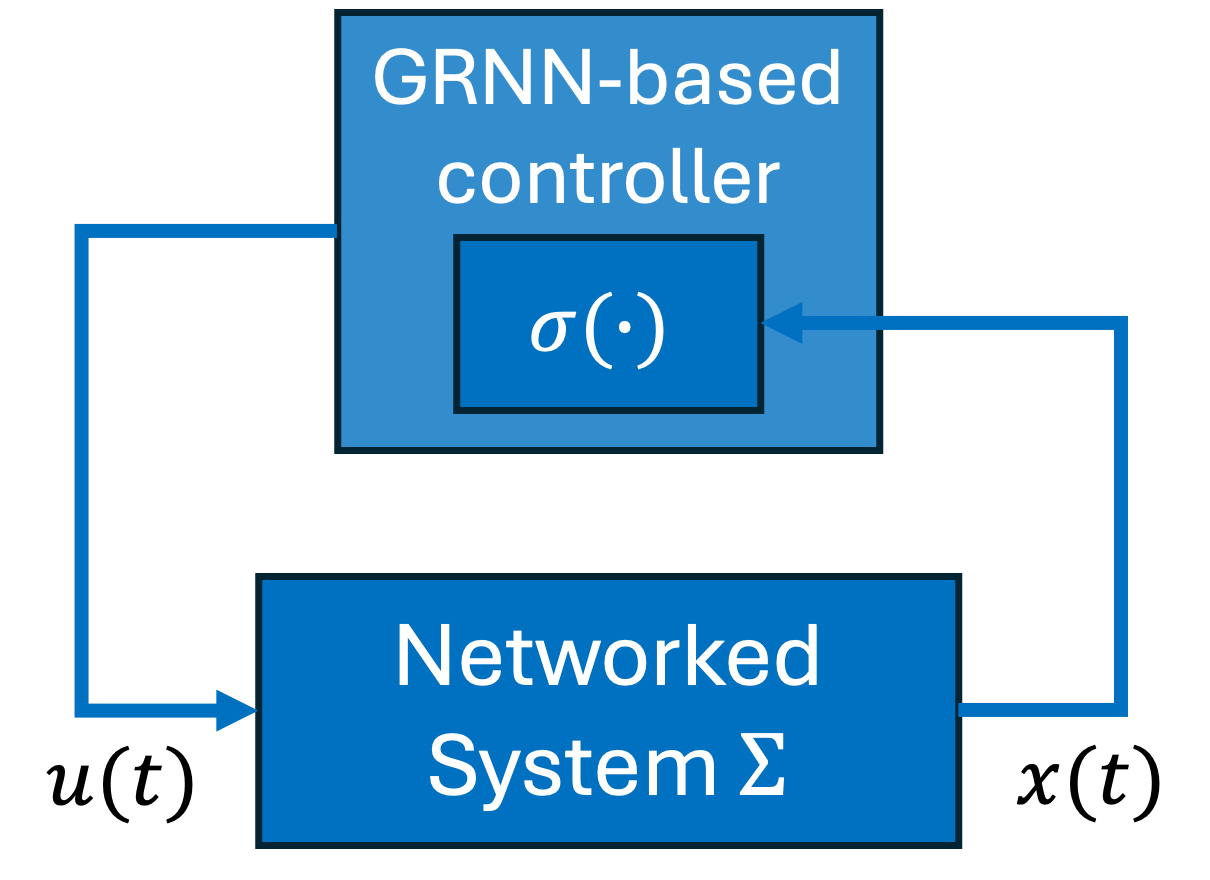}
    \caption{Architecture of the closed-loop networked system.}
    \label{Fig:closed-loop_architecture}
\end{figure}

In this way, the nonlinear activation function $\sigma(\cdot)$ can be separated, leading to a reformulation of the GRNN-based distributed optimal controller as
\begin{equation}
    \bmtx{u(t) \\ h(t)} = 
    \bmtx{\0 & K_4 & \0 \\ 
          (K_2+K_3 S) & \0 & K_1}
          \bmtx{x(t) \\ z(t) \\ z(t-1)},
\end{equation}
where $z(t) = \sigma(h(t))$.

Without loss of generality, it is assumed in our analysis below that the equilibrium point of the networked system \eqref{Eq:stacked_networked_system_rewritten} is $(x^*,\ u^*)=(\0,\ \0)$, and the corresponding hidden state and internal states of the GRNN-based controller at this point are $(h^*,\ z^*)=(\0,\ \0)$. 
Besides, for the hidden states $h$ on some interval $[\underline{h},\ \bar{h}]$ (element-wise with $\underline{h}\preceq \0\preceq \bar{h}$), the nonlinearity $\sigma$ of our GRNN-based controllers is locally sector-bounded in $[\underline{\alpha},\ \bar{\alpha}]$ and locally slope-restricted in $[\underline{\sigma},\ \bar{\sigma}]$ around the operating point $(h^*,\ z^*)=(\0,\ \0)$. 
For other cases when $(x^*,\ u^*,\ h^*,\ z^*)\neq(\0,\ \0,\ \0,\ \0)$, the analysis is easily generalized by considering the offset dynamics with the offset states $\tilde{x}:=x-x^*$, offset inputs $\tilde{u}:=u-u^*$, offset hidden states $\tilde{h}:=h-h^*$, and offset internal states $\tilde{z}:=z-z^*$.

Before we provide the proof of closed-loop stability, we need to obtain the bounds for the nonlinearity $\sigma$ in our proposed GRNN-based distributed optimal controller \eqref{Eq:closed-loop_GRNN_feedback} to facilitate the Lyapunov-based analysis.
With Def. \ref{Def:time_domain_IQC}, we can combine multiple IQCs, e.g., local sector boundedness and slope restriction (as seen in Def. \ref{def:local_sector-bounded_nonlinearity} and \ref{def:local_slope-restricted_nonlinearity}, respectively), for the sequences $h$ abd $z$, and we can readily construct the IQCs for both local sector and slope constraints as shown in the following lemma.
\begin{lemma}\label{Lem:sector_acausal_zames_falb_IQCs}
    Assume that a nonlinear element-wise function $\sigma:\mathbb{R}^{pN}\rightarrow\mathbb{R}^{pN}$ as in \eqref{Eq:closed-loop_GRNN_feedback} (element-wise on the interval $\underline{h}\preceq h\preceq \bar{h}$) is both locally sector-bounded in $[\underline{\alpha},\ \bar{\alpha}]$ (characterized by the IQC ($\Psi_{sec},\mathcal{Q}_{sec}$)) and locally slope-restricted in $[\underline{\sigma},\ \bar{\sigma}]$ around the point $(h^*,\ z^*)=(\0,\ \0)$ (characterized by the IQC ($\Psi_{slop},\ \mathcal{Q}_{slop}$)) element-wise.
    Then, the sequences $h$, $z\in\ell_2^{pN}$ satisfy the IQC defined by ($\Psi,\ \mathcal{Q}$) with 
    \begin{subequations}
        \begin{align}\label{Eq:Psi_IQC}
            \Psi:= \bmtx{\Psi_{sec} \\ \Psi_{slop}} =\scriptsize\bmtx{\begin{array}{c|c c} \0 & B^h_{\Psi_{slop}} & B^z_{\Psi_{slop}} \\  \hline 
            \bmtx{C_{\Psi_{sec}} \\ C_{\Psi_{slop}}} & \bmtx{D^h_{\Psi_{sec}} \\ D^h_{\Psi_{slop}}} & \bmtx{D^z_{\Psi_{sec}} \\ D^z_{\Psi_{slop}}}
            \end{array}},
        \end{align}
        \begin{align}\label{Eq:Q_IQC}
            \mathcal{Q}:=\scriptsize\bigg\{Q=\bmtx{Q_{sec} & \0 \\ \0 & Q_{slop}} : \begin{array}{c}
               Q_{sec}\in\mathcal{Q}_{sec}   \\
               Q_{slop}\in\mathcal{Q}_{slop} 
            \end{array}\bigg\},
        \end{align}
    \end{subequations}
    where $B^h_{\Psi_{slop}}=\scriptsize\bmtx{-\mbox{diag}(\bar{\sigma}) & \mbox{diag}(\underline{\sigma})}^\T$, $B^z_{\Psi_{slop}} = \scriptsize\bmtx{ \I & -\I}^\T$, $C_{\Psi_{sec}}=\scriptsize\bmtx{ \0 & \0}^\T$, $C_{\Psi_{slop}}=\scriptsize\bmtx{ \0 & \0 & \I}^\T$, $D^h_{\Psi_{sec}}=\scriptsize\bmtx{\mbox{diag}(\bar{\alpha}) & -\mbox{diag}(\underline{\alpha})}^\T$, $D^h_{\Psi_{slop}}=\scriptsize\bmtx{\mbox{diag}(\bar{\sigma}) \\ -\mbox{diag}(\underline{\sigma}) \\ \0}$, 
    $D^z_{\Psi_{sec}} = \scriptsize\bmtx{ -\I & \I }^\T$, and 
    $D^z_{\Psi_{slop}} = \scriptsize\bmtx{ -\I & \I & \0}^\T$. 
    Here, we select $Q_{sec} = Q_{off}(\mu, \mu)$ and  
    % for $\mu\in\mathbb{R}^n_+$, 
    $Q_{slop}=\scriptsize\bmtx{Q_{off}(\eta_0, \eta_0) & Q_{off}(\bar{\eta}, \underline{\eta}) \\
    Q_{off}(\underline{\eta}, \bar{\eta}) & \0}$ 
    % for $\eta_0$, $\underline{\eta}$, $\bar{\eta}\in\mathbb{R}^n_+$, 
    ($\eta_0\succeq\underline{\eta}+\bar{\eta}$) with $Q_{off}(a,b):=\scriptsize\bmtx{ \0 & \mbox{diag}(a) \\ \mbox{diag}(b) & \0}$, for $a$, $b\in\mathbb{R}^{pN}_+$.
    % $h_i^*\in[\underline{\alpha},\ \bar{\alpha}]$, 
\end{lemma}
\begin{proof}
    For any $h\in\mathbb{R}^{pN}$ and $z=\sigma(h)$, substituting the $\Psi$ and $\mathcal{Q}$ in \eqref{Eq:Psi_IQC} and \eqref{Eq:Q_IQC}, respectively, into the time domain IQC defined in Def. \ref{Def:time_domain_IQC}, we have:
    % \begin{align}
    %     \psi(t+1) &= \bmtx{-\mbox{diag}(\bar{\sigma}) \\ \mbox{diag}(\underline{\sigma})} h(t)+ \bmtx{\I \\ -\I} z(t),     \label{Eq:Psi_t+1}      \\
    %     q(t) &= \scriptsize\bmtx{\0 \\ \0 \\ \hline \0 \\ \0 \\ \I} \psi(t)+ \bmtx{\mbox{diag}(\bar{\alpha}) \\ -\mbox{diag}(\underline{\alpha}) \\ \hline \mbox{diag}(\bar{\sigma}) \\ -\mbox{diag}(\underline{\sigma}) \\ \0} h(t)+ \bmtx{ -\I \\ \I \\ \hline -\I \\ \I \\ \0} z(t)    \label{Eq:qt_compute} 
    % \end{align}

    \begin{align}
        & \psi(t+1) = \bmtx{-\mbox{diag}(\bar{\sigma}) \\ \mbox{diag}(\underline{\sigma})} h(t)+ \bmtx{\I \\ -\I} z(t),     \label{Eq:Psi_t+1}      \\
        & q(t) = \scriptsize\bmtx{\0 \\ \0 \\ \hline \0 \\ \0 \\ \I_{2pN\times pN}} \psi(t)+ \bmtx{\mbox{diag}(\bar{\alpha}) \\ -\mbox{diag}(\underline{\alpha}) \\ \hline \mbox{diag}(\bar{\sigma}) \\ -\mbox{diag}(\underline{\sigma}) \\ \0_{2pN\times pN}} h(t)+ \bmtx{ -\I \\ \I \\ \hline -\I \\ \I \\ \0_{2pN\times pN}} z(t)    \label{Eq:qt_compute} 
    \end{align}    

    Substituting \eqref{Eq:Psi_t+1} into \eqref{Eq:qt_compute}, we have:
    \begin{align}\label{Eq:constraint_output_q}
        q(t) =& \scriptsize\bmtx{\0 \\ \0 \\ \hline \0 \\ \0 \\ -\mbox{diag}(\bar{\sigma}) \\ \mbox{diag}(\underline{\sigma})}  h(t-1)+ \bmtx{\0 \\ \0 \\ \hline \0 \\ \0 \\ \I \\ -\I}z(t-1)+    \nonumber   \\
        &\scriptsize\bmtx{\mbox{diag}(\bar{\alpha}) \\ -\mbox{diag}(\underline{\alpha}) \\ \hline \mbox{diag}(\bar{\sigma}) \\ -\mbox{diag}(\underline{\sigma}) \\ \0} h(t)+ \bmtx{ -\I \\ \I \\ \hline -\I \\ \I \\ \0} z(t) \nonumber    \\
        =&\scriptsize\bmtx{\mbox{diag}(\bar{\alpha})h(t)-z(t) \\ -\mbox{diag}(\underline{\alpha})h(t)+z(t) \\ \hline \mbox{diag}(\bar{\sigma})h(t)-z(t) \\ -\mbox{diag}(\underline{\sigma})h(t)+z(t) \\ -\mbox{diag}(\bar{\sigma})h(t-1)+z(t-1) \\ \mbox{diag}(\underline{\sigma})h(t-1)-z(t-1)}
    \end{align}

    % If for $Q\in\mathcal{Q}$, and $\tau\in\mathbb{N}$, the following condition holds
    % \begin{equation}
    %     \sum_{\tau=0}^{t} q^\T(\tau)Qq(\tau)\geq 0,
    % \end{equation}
    Then, we have:
    \begin{align}
        & \scriptsize \sum_{\tau=0}^{t} q^\T(\tau)\bmtx{\begin{array}{c|c c} Q_{off}(\mu,\mu) & \0 & \0 \\ \hline \0 & Q_{off}(\eta_0,\eta_0) & Q_{off}(\bar{\eta},\underline{\eta}) \\
        \0 & Q_{off}(\underline{\eta},\bar{\eta}) & \0
        \end{array}} q(\tau)     \nonumber   \\
        =& \scriptsize\sum_{\tau=0}^{t} \bmtx{\mbox{diag}(\bar{\alpha})h(\tau)-z(\tau) \\ -\mbox{diag}(\underline{\alpha})h(\tau)+z(\tau)}^\T Q_{off}(\mu,\mu)\bmtx{\mbox{diag}(\bar{\alpha})h(\tau)-z(\tau) \\ -\mbox{diag}(\underline{\alpha})h(\tau)+z(\tau)}+    \nonumber  \\
        & \scriptsize\bmtx{\mbox{diag}(\bar{\sigma})h(\tau)-z(\tau) \\ -\mbox{diag}(\underline{\sigma})h(\tau)+z(\tau) \\ -\mbox{diag}(\bar{\sigma})h(\tau-1)+z(\tau-1) \\ \mbox{diag}(\underline{\sigma})h(\tau-1)-z(\tau-1)}^\T\scriptsize\bmtx{Q_{off}(\eta_0, \eta_0) & Q_{off}(\bar{\eta}, \underline{\eta}) \\
        Q_{off}(\underline{\eta}, \bar{\eta}) & \0}*  \nonumber  \\
        & \scriptsize\bmtx{\mbox{diag}(\bar{\sigma})h(\tau)-z(\tau) \\ -\mbox{diag}(\underline{\sigma})h(\tau)+z(\tau) \\ -\mbox{diag}(\bar{\sigma})h(\tau-1)+z(\tau-1) \\ \mbox{diag}(\underline{\sigma})h(\tau-1)-z(\tau-1)} \geq 0  \label{Eq:qt_T_Q_qt_compute}   
        % =&\scriptsize 2(\mbox{diag}(\bar{\alpha})h(t)-z(t))^\T\mbox{diag}(\mu)(-\mbox{diag}(\underline{\alpha})h(t)+z(t)) +     \nonumber    \\
        % & \scriptsize 2(\mbox{diag}(\bar{\beta})h(t)-z(t))^\T\mbox{diag}(\eta_0)(-\mbox{diag}(\bar{\beta})h(t)+z(t))+        \nonumber    \\
        % & \scriptsize 2(\mbox{diag}(\bar{\beta})h(t)-z(t))^\T\mbox{diag}(\bar{\eta})(\mbox{diag}(\bar{\beta})h(t-1)-z(t-1))+        \nonumber    \\
        % & \scriptsize 2(\mbox{diag}(\bar{\beta})h(t)-z(t))^\T\mbox{diag}(\underline{\eta})(-\mbox{diag}(\bar{\beta})h(t-1)+z(t-1)).   
    \end{align}    

    Since $\mu\in\mathbb{R}^{pN}_+$, $\eta_0$, $\underline{\eta}$, $\bar{\eta}\in\mathbb{R}^{2pN}_+$ and $\eta_0\succeq\underline{\eta}+\bar{\eta}$, it is readily shown that both terms in \eqref{Eq:qt_T_Q_qt_compute} are positive semi-definite, which indicates the element-wise satisfaction of both local sector-boundedness in $[\underline{\alpha},\ \bar{\alpha}]$ and local slope-restriction in $[\underline{\sigma},\ \bar{\sigma}]$ around the point $(h^*,\ z^*)=(\0,\ \0)$ for the nonlinear activation function $\sigma(h)$ on the interval $\underline{h}\preceq h\preceq \bar{h}$.
\end{proof}

\begin{remark}
    The results in Lem. \ref{Lem:sector_acausal_zames_falb_IQCs} seems to be similar as those in [Lems. 5 and 6 \cite{wu2022stability}]. 
    However, it is worth noting that the multi-layer RNN-based feedback controller in \cite{wu2022stability} is only to control a single LTI system, and the results in Lem. \ref{Lem:sector_acausal_zames_falb_IQCs} are for networked systems and the GRNN-based controllers in \eqref{Eq:distributed_optimal_control_GRNN_local} are distributed with only one layer.
\end{remark}

Our stability analysis will rely on the element-wise satisfaction of local sector boundedness and local slope-restriction of the the nonlinear activation function $\sigma(\cdot)$ in \eqref{Eq:closed-loop_GRNN_feedback}, and we summarize these constraints in the following property based on the results in Lem. \ref{Lem:sector_acausal_zames_falb_IQCs}.

% {\color{red} we stop at here yesterday!!! We need to finish the analysis of the stability for closed-loop dynamics!}

\begin{property}\label{proper:local_sector_slope_constraints}
    Consider the operating point $(h^*,\ z^*)=(\0,\ \0)$ of the nonlinear activation function $\sigma$ in \eqref{Eq:closed-loop_GRNN_feedback}. Suppose that the input bound of $\sigma$ is obtained as $\underline{h}\preceq h \preceq\bar{h}$ for every element in $h$.
    There exist $\underline{\alpha}$, $\bar{\alpha}$, $\underline{\sigma}$ and $\bar{\sigma}$ such that the nonlinear activation function $\sigma$ in \eqref{Eq:closed-loop_GRNN_feedback} is both local sector-bounded on the interval $[\underline{\alpha},\bar{\alpha}]$ and local slope-restricted around the origin in $[\underline{\sigma},\bar{\sigma}]$ are guaranteed.
\end{property}

\begin{remark}
    In order to apply the local sector and slope constraints on the nonlinear activation function $\sigma$ in \eqref{Eq:closed-loop_GRNN_feedback} in the stability analysis, we must first compute the bounds $\underline{h}$, $\bar{h}\in\mathbb{R}^{pN}$ for the activation function input $h\in\mathbb{R}^{pN}$. The process to compute the bounds is out of the scope of this paper but more details can be found in \cite{gowal2018effectiveness}.
\end{remark}

Define a stacked vector $\xi(t):=[x^\T(t),\ h^\T(t-1),\ z^\T(t-1)]^\T\in\mathbb{R}^{nN+2pN}$.
Note that the closed-loop networked system under our GRNN-based distributed optimal controller \eqref{Eq:closed-loop_GRNN_feedback} with the constrained output $q(t)$ (as in \eqref{Eq:constraint_output_q}) of the nonlinear activation function $\sigma$ satisfying the IQC defined by ($\Psi,\ \mathcal{Q}$) (as in \eqref{Eq:Psi_IQC} and \eqref{Eq:Q_IQC}) can be constructed as an augmented system as:
    \begin{subequations}\label{Eq:augmented_closed-loop_system}
        \begin{align}
            \xi(t+1) &= A_{\Xi}\xi(t)+B_{\Xi}\bmtx{h^\T(t) & z^\T(t)}^\T,    \\
            q(t) &= C_{\Xi}\xi(t)+D_{\Xi}\bmtx{h^\T(t) & z^\T(t)}^\T,
        \end{align}
    \end{subequations}
where $A_{\Xi}:=\scriptsize\bmtx{A & \0 & \0 \\ (K_2+K_3 S) & \0 & K_1 \\ \0 & \0 & \0}$, $B_{\Xi}:=\scriptsize\bmtx{\0 & BK_4 \\ \0 & \0 \\ \0 & \I}$, $C_{\Xi}:=\scriptsize\bmtx{\begin{array}{c|c c}\0_{4pN\times nN} & \0_{4pN\times pN} & \0_{4pN\times pN} \\ \hline \0 & -\mbox{diag}(\bar{\sigma}) & \I \\ \0 & \mbox{diag}(\underline{\sigma}) & -\I \end{array}}$, and 
$D_{\Xi}:=\scriptsize\bmtx{\mbox{diag}(\bar{\alpha}) & -\I \\ -\mbox{diag}(\underline{\alpha}) & \I \\ \hline \mbox{diag}(\bar{\sigma}) & -\I \\ -\mbox{diag}(\underline{\sigma}) & \I \\ \0_{2pN\times pN} & \0_{2pN\times pN}}$.

Now, using Property \ref{proper:local_sector_slope_constraints}, we show the closed-loop stability for the networked system \eqref{Eq:augmented_closed-loop_system} (which is the equivalent combination of the networked system \eqref{Eq:closed-loop_GRNN_feedback} and the constrained output of the nonlinearity $\sigma$) in the following theorem.
\begin{theorem}
    Consider the closed-loop networked system under the GRNN-based distributed optimal controller as in \eqref{Eq:augmented_closed-loop_system},
    where the nonlinear activation function $\sigma$ satisfies the Property \ref{proper:local_sector_slope_constraints}.
    % are both element-wise local sector-bounded in $[\underline{\alpha},\ \bar{\alpha}]$ and slope-restricted in $[\underline{\sigma},\ \bar{\sigma}]$ on some interval $\underline{h}\preceq h\preceq\bar{h}$ around the origin.
    If there exist $P>0$ and $\epsilon>0$ such that the condition
    \begin{align}
        \scriptsize\bmtx{
        A_{\Xi}^\T PA_{\Xi}-P+C_{\Xi}^\T QC_{\Xi} & A_{\Xi}^\T PB_{\Xi}+C_{\Xi}^\T QD_{\Xi} \\ 
        B_{\Xi}^\T PA_{\Xi}+D_{\Xi}^\T QC_{\Xi}   & B_{\Xi}^\T PB_{\Xi}+D_{\Xi}^\T QD_{\Xi}}\leq \bmtx{-\epsilon\I_{nN\times nN} & \0 \\
        \0 & \0},   \label{Eq:main_stability_condition}
    \end{align}
    holds, then the closed-loop networked system \eqref{Eq:augmented_closed-loop_system} is local asymptotically stable.  
\end{theorem}

\begin{proof}
    Select the Lyapunov function candidate $V(\xi):=\xi^\T(t) P\xi(t)$. For \eqref{Eq:main_stability_condition}, if we left/right multiply the vector $v:=[\xi^\T\ h^\T\ z^\T]$ and its transpose, then we have:
    \begin{align}
        & v^\T\bmtx{
        A_{\Xi}^\T PA_{\Xi}-P+C_{\Xi}^\T QC_{\Xi} & A_{\Xi}^\T PB_{\Xi}+C_{\Xi}^\T QD_{\Xi} \\ 
        B_{\Xi}^\T PA_{\Xi}+D_{\Xi}^\T QC_{\Xi}   & B_{\Xi}^\T PB_{\Xi}+D_{\Xi}^\T QD_{\Xi}}v\leq  \nonumber \\ 
        & v^\T\bmtx{-\epsilon\I_{nN\times nN} & \0 \\
        \0 & \0}v,  
    \end{align}
    which indicates that 
    \begin{align}
        & V(t+1)-V(t)+q^\T(t) Qq(t)    \nonumber   \\
        =&\xi^\T(t+1) P\xi(t+1)-\xi^\T(t) P\xi(t)+q^\T(t) Qq(t)      \nonumber   \\
        % =&\bmtx{\xi \\ \hline h \\ z}^\T\bmtx{\begin{array}{c|c}
        % A_{\Xi}^\T PA_{\Xi}-P & A_{\Xi}^\T PB_{\Xi} \\  \hline
        % B_{\Xi}^\T PA_{\Xi}   & B_{\Xi}^\T PB_{\Xi}
        % \end{array}}\bmtx{\xi \\ \hline h \\ z}+ \bmtx{\xi \\ \hline h \\ z}^\T\bmtx{\begin{array}{c|c}
        % C_{\Xi}^\T QC_{\Xi} & C_{\Xi}^\T QD_{\Xi} \\  \hline
        % D_{\Xi}^\T QC_{\Xi}   & D_{\Xi}^\T QD_{\Xi}
        % \end{array}}\bmtx{\xi \\ \hline h \\ z}  \nonumber \\
        =& \scriptsize\bmtx{\xi \\ \hline h \\ z}^\T\bmtx{\begin{array}{c|c}
        A_{\Xi}^\T PA_{\Xi}-P+C_{\Xi}^\T QC_{\Xi} & A_{\Xi}^\T PB_{\Xi}+C_{\Xi}^\T QD_{\Xi} \\  \hline
        B_{\Xi}^\T PA_{\Xi}+D_{\Xi}^\T QC_{\Xi}   & B_{\Xi}^\T PB_{\Xi}+D_{\Xi}^\T QD_{\Xi}
        \end{array}}\bmtx{\xi \\ \hline h \\ z}   \nonumber \\
        \leq& -\epsilon|x|^2.
    \end{align}

    Therefore, if the condition \eqref{Eq:main_stability_condition} holds, then we have:
    \begin{equation}
        V(t+1)-V(t)+\underbrace{q^\T(t) Qq(t)}_{\geq 0}\leq -\epsilon|x|^2,
    \end{equation}
    which indicates that for any initial state around the origin, e.g., $x(0)\in\mathcal{E}(P_x)\triangleq\{x^\T P_x x\leq 1: P_x>0\}$, where the ellipsoid $\mathcal{E}(P_x)$ is an approximation of the region of attraction for the state $x$, and $P_x$ is the upper left block of the block matrix $P\triangleq\scriptsize\bmtx{P_x & \0 \\ \0 & *}$, the system state $x(t)\rightarrow \0$ as $t\rightarrow\infty$, based on LaSalle’s Invariance principle \cite{mei2017lasalle}.
    This completes the proof.

\end{proof}

\section{Simulation Example}\label{sec:simulation}

In this section, we provide a simulation example to verify the effectiveness of our proposed methods.
We consider a classical distributed LQR control problem as follows:
\begin{subequations}\label{Eq:centralized_LQR_problem}
    \begin{align}
        \min_{u}\ \ & \mathbb{E}_{x(0)\sim\mathcal{N}(\0,\I)}\Big(\sum_{t=0}^{\infty}\ x^\T(t)P_tx(t)+u^\T(t)R_t u(t)+   \nonumber  \\
        & \ \ \ \ \ \ \ \ \ \ \ \ \ \ \ \ \ \ \ \ \ \ \ \
        \ \ \ \ \ \ \ \ \ \ \ \ \ x^\T(T)P_Tx(T)\Big)  \\
        s.t.\ & x(t+1) = Ax(t)+Bu(t),\ t=0,1,...  \\
        & u(t)=\phi(x(t);G),\ \mbox{for }\  t=0,1,...
    \end{align}
\end{subequations}
We use the GNN framework to solve it by setting $X(t) = x(t)$ and $U(t) = u(t)$. Here, the parameters are selected as $P_t=R_t=P_T=\mathbf{I}$. Therefore, by using our proposed GRNN-based control method with online distributed training, we can reformulate the problem in \eqref{Eq:centralized_LQR_problem} as
\begin{subequations}\label{Eq:distributed_LQR_problem}
    \begin{align}
        \min_{u}\ \ & \mathbb{E}_{x(0)\sim\mathcal{N}(\0,\I)}\Big(\sum_{i=1}^N\sum_{t=0}^{\infty}\ x_i^\T(t)P_{ti}x_{i}(t)+   \nonumber  \\
        & \ \ \ \ \ \ \ \ \ \ \ u_i^\T(t)R_{ti} u_i(t)+x_i^\T(T)P_{Ti}x_i(T)\Big)  \\
        s.t.\ & \eqref{Eq:ith_subsystem},\ \eqref{Eq:distributed_optimal_control_GRNN_local},\ \forall i\in\mathcal{I}_N,\ j\in\mathcal{N}_i, \mbox{and } \forall t\in\{0,1,...\},
    \end{align}
\end{subequations}
% and the $P$ matrix is determined by solving the Discrete Algebraic Ricatti Equation (DARE):
% \begin{equation}
%     A^\T PA-P-A^\T PB(B^\T PB+R)^{-1}B^\T PA=-Q,
% \end{equation}
% which guarantees the stability of the solution. 

In this way, our proposed results are verified by a specifically developed simulator\footnote{Publicly available at \href{https://github.com/NDzsong2/fdTrainGRNN\textunderscore NetCtrl-main.git}{https://github.com/NDzsong2/fdTrainGRNN\textunderscore NetCtrl-main.git}}, which is constructed in Python and TensorFlow environment. Our GRNN-based controller training is conducted in a self-supervised manner.
We assume that the networked system involves $10$ subsystems, and each is with $2$ dimensions. 
The topology $S$ of the network is given by randomly generating the graph based on the Gaussian Random Partition Model (GRPM).
The network dynamics is randomly created using NumPy, and each value is sampled from $\mathcal{N}(0,1)$ (standard normal distribution), where the sparsity of the $A$ and $B$ follow the same pattern as the topology $S$.
Besides, we ensure the controllability of the pair ($A,B$). Similar to \cite{gama2022distributed}, we normalized the $A$ and $B$ matrices and then scaled with some value $0.995$.
A random noise $n\sim\mathcal{N}(\0,\ 0.1\I)$ is also added to the networked system dynamics.
We select the total training and testing epochs as $21$. The batch size and the test samples are selected as $100$ and $20$, respectively. The dimension of the hidden states $z_{i}$ and the internal states $h_i$ are both $\mathbb{R}^{2\times 2}$, and the nonlinear activation function $\sigma$ is selected as $\tanh$.
Initially, we randomly select $x(0)\sim\mathcal{N}(2\I,\ \I)$, $z(-1)=\0$, and $\Theta_{ki}\sim U(\0,\ \I)$ (uniform distribution), for all $k\in\mathcal{I}_4$, and $i\in\mathcal{I}_N$.
The activation function $\sigma$ in our proposed GRNN-based controller \eqref{Eq:distributed_optimal_control_GRNN_local} is selected as $\tanh$ function.

\begin{figure}[!t]
    \vspace{2mm}
    \centering
    \begin{subfigure}{0.5\textwidth}
        \includegraphics[width=\linewidth]{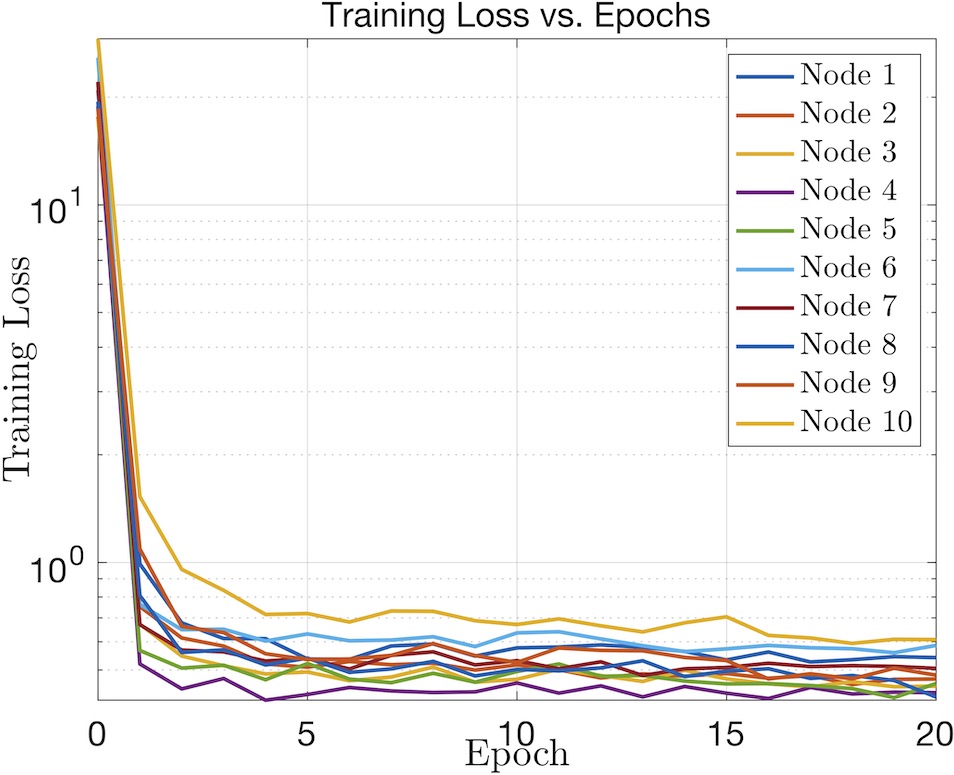}
        \vspace{-3.5mm}
        \caption{Training loss of our proposed GRNN-based controller}
        \label{Fig:training_loss_our_grnn}
    \end{subfigure}
    \hfill
    \begin{subfigure}{0.5\textwidth}
        \includegraphics[width=\linewidth]{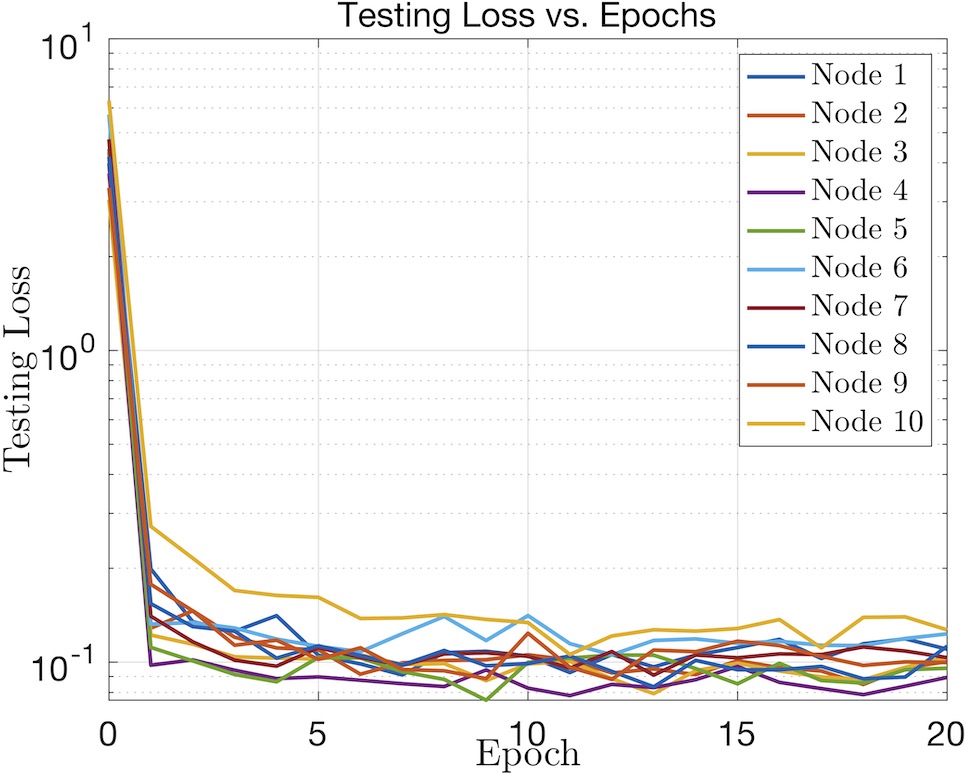}
        \vspace{-3.5mm}
        \caption{Testing loss of our proposed GRNN-based controller}
        \label{Fig:testing_loss_our_grnn}
    \end{subfigure}
    \caption{Results observed under our proposed GRNN-based controller: (a) Training loss; (b) testing loss.}
    \label{Fig:training_testing_results_our_grnn_method}
\end{figure}

\begin{figure}[!t]
    \vspace{2mm}
    \centering
    \begin{subfigure}{0.5\textwidth}
        \includegraphics[width=\linewidth]{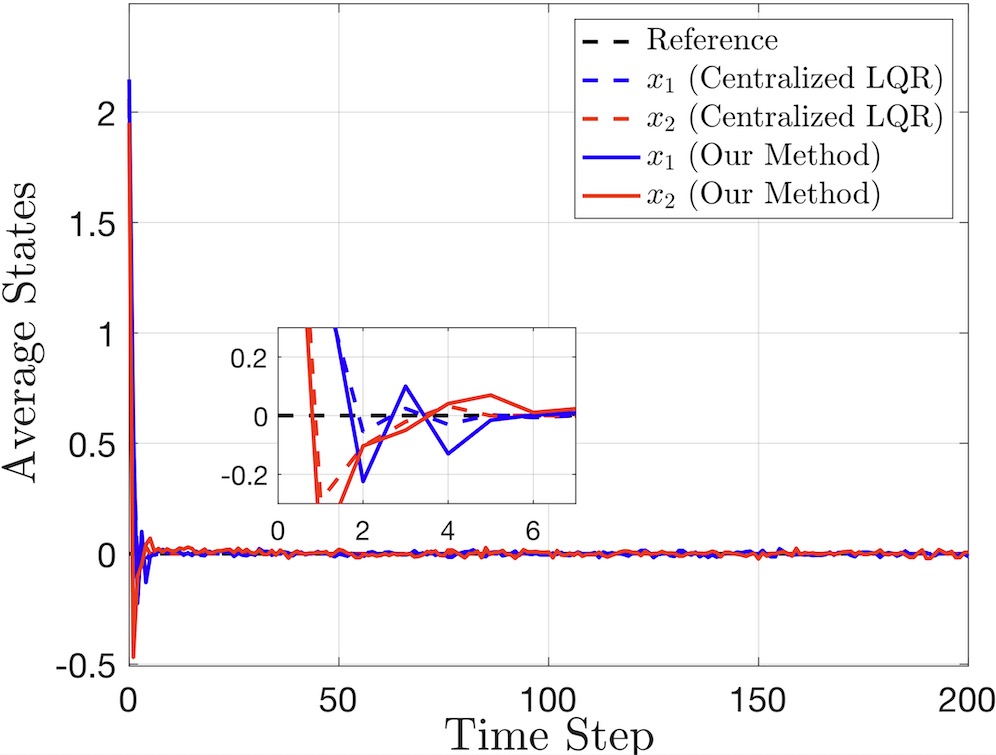}
        \vspace{-3.5mm}
        \caption{Average states in real time}
        \label{Fig:training_loss_our_grnn}
    \end{subfigure}
    \hfill
    \begin{subfigure}{0.5\textwidth}
        \includegraphics[width=\linewidth]{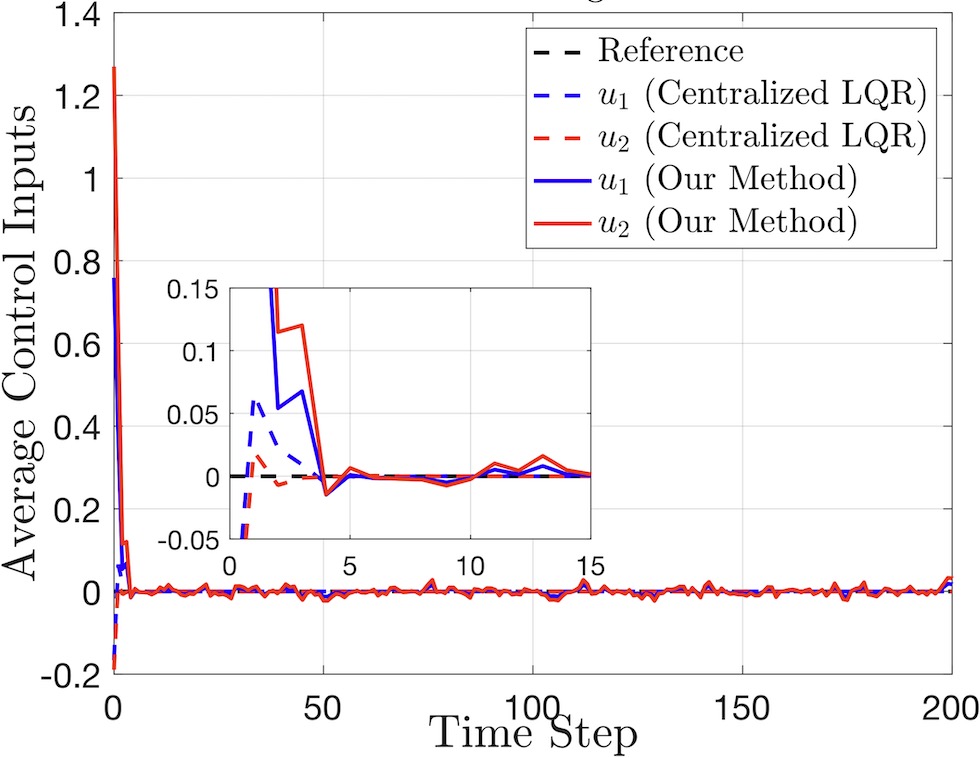}
        \vspace{-3.5mm}
        \caption{Average control inputs in real time}
        \label{Fig:testing_loss_our_grnn}
    \end{subfigure}
    \caption{Comparison between our proposed method and centralized LQR (global optimal) solutions.}
    \label{Fig:average_states_controls_comparisons}
\end{figure}

The training and testing results of our proposed GRNN-based controller are shown in Fig. \ref{Fig:training_loss_our_grnn} and \ref{Fig:testing_loss_our_grnn}, respectively, using the semilogy plots. It is observed that both training and testing losses for all subsystems can rapidly decrease from the initial guess and stay around $0.5$ and $0.1$, respectively. The reason for the deviation from $0$ is due to the noise ($n\sim\mathcal{N}(\0,\ 0.1\I)$) added in the network dynamics.

To illustrate the capability of online training and testing, we also count the time of the entire training and testing process in our simulation.
In particular, the total running time of our training and testing process is $4.10$s, and the average training and testing time for each node in each epoch is only $0.02$s.
Besides, as our network size $N$ grows, our computational complexity is $\mathcal{O}(N)$ since each individual controller only contains $4$ trainable weights as seen in \eqref{Eq:distributed_optimal_control_GRNN_local}.
Therefore, this indicates that our proposed GRNN-based controller can be trained and tested online for large-scale networked systems.

We further illustrate the performance of our proposed method by comparing to the classical centralized LQR solutions as shown in Fig. \ref{Fig:average_states_controls_comparisons}. 
It is observed that our proposed GRNN method can eventually track the global optimal LQR controller, implying the optimality of our controller. The solving time of the centralized LQR controller is around $0.015$s, which also shows the scalability of our proposed method. Despite the better performance, centralized LQR solution requires global information for the implementation, and the controllers may be redesigned when the network changes.

%---------------------------------------------------------------
\section{Conclusion}\label{sec:conclusion}

In this paper, we proposed a GRNN-based distributed optimal control framework that enables online distributed training. The distributed gradient computation during back-propagation was first illustrated using the chain rule. The weights of GRNN were updated using the idea of (consensus-based) distributed optimization. Then, the closed-loop stability was proved for the networked system under our proposed GRNN-based controller by assuming that the output of the nonlinear activation function is both local sector-bounded and slope-restricted.
A simulator was specifically developed to verify the effectiveness of our proposed GRNN-based controller.
Future works will generalize our proposed control framework to more general nonlinear networked systems and consider the design of the interconnection topologies.

% %---------------------------------------------------------------
% \section{Appendix}

% \begin{proof}
    
% \end{proof}

%---------------------------------------------------------------
\bibliographystyle{IEEEtran}
\bibliography{references}

\end{document}